\documentstyle[fleqn,times,twocolumn]{article}

\def\istcsabs{1}

\ifnum\istcsabs=1

\else

\fi

\newlength{\saveparindent}
\newlength{\saveparskip}
\newcommand{\eqref}[1]{\mbox{Equation~(\ref{#1})}}

\newcounter{ctr}

\makeatletter

\ifnum\istcsabs=0
\def\subsubsection{\@startsection{subsubsection}{3}{\z@}{-3.25ex plus
-1ex minus -.2ex}{1.5ex plus .2ex}{\large\rm}}

\else
\def\subsubsection{\@startsection{subsubsection}{3}{\z@}{-1.25ex plus
-1ex minus -.2ex}{1.5ex plus .2ex}{\large\rm}}

\fi

\ifnum\istcsabs=0

\def\appearsin#1{\gdef\@appearsin{#1}}

\def\maketitle{\par
 \begingroup
 \def\thefootnote{\fnsymbol{footnote}}
 \def\@makefnmark{\hbox
 to 0pt{$^{\@thefnmark}$\hss}}
 \if@twocolumn
 \twocolumn[\@maketitle]
 \else \newpage
 \global\@topnum\z@ \@maketitle \fi\thispagestyle{plain}\@thanks
 \endgroup
 \setcounter{footnote}{0}
 \let\maketitle\relax
 \let\@maketitle\relax
 \gdef\@thanks{}\gdef\@author{}\gdef\@title{}\gdef\@appearsin{}
          \let\thanks\relax}
\def\@maketitle{\newpage
 \noindent \@appearsin
 \vskip 1in \begin{center}
 {\LARGE \@title \par} \vskip 1.5em {\large \lineskip .5em
\begin{tabular}[t]{c}\@author
 \end{tabular}\par}
 \vskip 1em {\normalsize \@date} \end{center}
 \par
 \vskip 1.5em}
\def\abstract{\if@twocolumn
\section*{Abstract}
\else \small
\begin{center}
{\bf Abstract\vspace{-.5em}\vspace{0pt}}
\end{center}
\quotation
\fi}
\def\endabstract{\if@twocolumn\else\endquotation\fi}
\mark{{}{}}

\fi

\newcommand{\stoc}[1]{\ifnum#1=
71{{\sl Proceedings of the Third Annual Symposium
on the Theory of Computing,\/} ACM, 1971}\else{\ifnum#1=
83{{\sl Proceedings of the Fifteenth Annual Symposium
on the Theory of Computing,\/} ACM, 1983}\else{\ifnum#1=
84{{\sl Proceedings of the Sixteenth Annual Symposium
on the Theory of Computing,\/} ACM, 1984}\else{\ifnum#1=
85{{\sl Proceedings of the Seventeenth Annual Symposium
on the Theory of Computing,\/} ACM, 1985}\else{\ifnum#1=
86{{\sl Proceedings of the Eighteenth Annual Symposium
on the Theory of Computing,\/} ACM, 1986}\else{\ifnum#1=
87{{\sl Proceedings of the Nineteenth Annual Symposium
on the Theory of Computing,\/} ACM, 1987}\else{\ifnum#1=
88{{\sl Proceedings of the Twentieth Annual Symposium
on the Theory of Computing,\/} ACM, 1988}\else{\ifnum#1=
89{{\sl Proceedings of the Twenty First Annual Symposium
on the Theory of Computing,\/} ACM, 1989}\else{\ifnum#1=
90{{\sl Proceedings of the Twenty Second Annual Symposium
on the Theory of Computing,\/} ACM, 1990}\else{\ifnum#1=
91{{\sl Proceedings of the Twenty Third Annual Symposium
on the Theory of Computing,\/} ACM, 1991}\else{\ifnum#1=
92{{\sl Proceedings of the Twenty Fourth Annual Symposium
on the Theory of Computing,\/} ACM, 1992}\else{\ifnum#1=
93{{\sl Proceedings of the Twenty Fifth Annual Symposium
on the Theory of Computing,\/} ACM, 1993}\else{\ifnum#1=
94{{\sl Proceedings of the Twenty Sixth Annual Symposium
on the Theory of Computing,\/} ACM, 1994}\else
This STOC not yet defined!
\fi}\fi}\fi}\fi}\fi}\fi}\fi}\fi}\fi}\fi}\fi}\fi}\fi}

\newcommand{\focs}[1]{\ifnum#1=
78{{\sl Proceedings of the Nineteenth Annual Symposium
on the Foundations of Computer Science,\/} IEEE, 1978}\else{\ifnum#1=
80{{\sl Proceedings of the Twenty First Annual Symposium
on the Foundations of Computer Science,\/} IEEE, 1980}\else{\ifnum#1=
82{{\sl Proceedings of the Twenty Third Annual Symposium
on the Foundations of Computer Science,\/} IEEE, 1982}\else{\ifnum#1=
85{{\sl Proceedings of the Twenty Sixth Annual Symposium
on the Foundations of Computer Science,\/} IEEE, 1985}\else{\ifnum#1=
86{{\sl Proceedings of the Twenty Seventh Annual Symposium
on the Foundations of Computer Science,\/} IEEE, 1986}\else{\ifnum#1=
87{{\sl Proceedings of the Twenty Eighth Annual Symposium
on the Foundations of Computer Science,\/} IEEE, 1987}\else{\ifnum#1=
88{{\sl Proceedings of the Twenty Ninth Annual Symposium
on the Foundations of Computer Science,\/} IEEE, 1988}\else{\ifnum#1=
89{{\sl Proceedings of the Thirtieth Annual Symposium
on the Foundations of Computer Science,\/} IEEE, 1989}\else{\ifnum#1=
90{{\sl Proceedings of the Thirty First Annual Symposium
on the Foundations of Computer Science,\/} IEEE, 1990}\else{\ifnum#1=
91{{\sl Proceedings of the Thirty Second Annual Symposium
on the Foundations of Computer Science,\/} IEEE, 1991}\else{\ifnum#1=
92{{\sl Proceedings of the Thirty Third Annual Symposium
on the Foundations of Computer Science,\/} IEEE, 1992}\else{\ifnum#1=
93{{\sl Proceedings of the Thirty Fourth Annual Symposium
on the Foundations of Computer Science,\/} IEEE, 1993}\else{\ifnum#1=
94{{\sl Proceedings of the Thirty Fifth Annual Symposium
on the Foundations of Computer Science,\/} IEEE, 1994}\else
This FOCS not yet defined!
\fi}\fi}\fi}\fi}\fi}\fi}\fi}\fi}\fi}\fi}\fi}\fi}\fi}

\newcommand{\crypto}[1]{\ifnum#1=
84{{\sl Advances in Cryptology -- Crypto~84 Proceedings,}
Lecture Notes in Computer Science Vol.~196, Springer-Verlag, B.~Blakley, ed.,
1985}\else{\ifnum#1=
85{{\sl Advances in Cryptology -- Crypto~85 Proceedings,}
Lecture Notes in Computer Science Vol.~218, Springer-Verlag, H.~Williams, ed.,
1985}\else{\ifnum#1=
86{{\sl Advances in Cryptology -- Crypto~86 Proceedings,}
Lecture Notes in Computer Science Vol.~263, Springer-Verlag, A.~Odlyzko, ed.,
1986}\else{\ifnum#1=
87{{\sl Advances in Cryptology -- Crypto~87 Proceedings,}
Lecture Notes in Computer Science Vol.~293, Springer-Verlag, C.~Pomerance, ed.,
1987}\else{\ifnum#1=
88{{\sl Advances in Cryptology -- Crypto~88 Proceedings,}
Lecture Notes in Computer Science Vol.~403, Springer-Verlag, S.~Goldwasser,
ed., 1989}\else{\ifnum#1=
89{{\sl Advances in Cryptology -- Crypto~89 Proceedings,}
Lecture Notes in Computer Science Vol.~435, Springer-Verlag, G.~Brassard, ed.,
1989}\else{\ifnum#1=
90{{\sl Advances in Cryptology -- Crypto~90 Proceedings,}
Lecture Notes in Computer Science Vol.~537, Springer-Verlag,
A.J.~Menezes and S.~Vanstone, ed., 1990}\else{\ifnum#1=
91{{\sl Advances in Cryptology -- Crypto~91 Proceedings,}
Lecture Notes in Computer Science Vol.~576, Springer-Verlag, J.~Feigenbaum,
ed., 1991}\else{\ifnum#1=
92{{\sl Advances in Cryptology -- Crypto~92 Proceedings,}
Lecture Notes in Computer Science Vol.~740, Springer-Verlag, E.~Brickell,
ed., 1993}\else{\ifnum#1=
93{{\sl Advances in Cryptology -- Crypto~93 Proceedings,}
Lecture Notes in Computer Science Vol.~773, Springer-Verlag, D.~Stinson,
ed., 1994}\else{\ifnum#1=
94{{\sl Advances in Cryptology -- Crypto~94 Proceedings}}\else
This Crypto not yet defined!
\fi}\fi}\fi}\fi}\fi}\fi}\fi}\fi}\fi}\fi}\fi}

\newcommand{\eurocrypt}[1]{\ifnum#1=
87{{\sl Advances in Cryptology -- Eurocrypt~87
Proceedings,} Lecture Notes in Computer Science Vol.~304, Springer-Verlag,
D.~Chaum, ed., 1987}\else{\ifnum#1=
88{{\sl Advances in Cryptology -- Eurocrypt~88
Proceedings,} Lecture Notes in Computer Science Vol.~330, Springer-Verlag,
C.~G.~Gunther, ed., 1988}\else{\ifnum#1=
89{{\sl Advances in Cryptology -- Eurocrypt~89
Proceedings,} Lecture Notes in Computer Science Vol.~434, Springer-Verlag,
J-J.~Quisquater, J.~Vandewille, ed., 1989}\else{\ifnum#1=
90{{\sl Advances in Cryptology -- Eurocrypt~90
Proceedings,} Lecture Notes in Computer Science Vol.~473, Springer-Verlag,
I.~Damg{\aa}rd, ed., 1990}\else{\ifnum#1=
91{{\sl Advances in Cryptology -- Eurocrypt~91
Proceedings,} Lecture Notes in Computer Science Vol.~547, Springer-Verlag
 1991}\else{\ifnum#1=
92{{\sl Advances in Cryptology -- Eurocrypt~92
Proceedings,} Lecture Notes in Computer Science Vol.~658, Springer-Verlag
 1992}\else{\ifnum#1=
93{{\sl Advances in Cryptology -- Eurocrypt~93
Proceedings}, 1993}\else{\ifnum#1=
94{{\sl Advances in Cryptology -- Eurocrypt~94
Proceedings}, 1994}\else
This Eurocrypt not yet defined!
\fi}\fi}\fi}\fi}\fi}\fi}\fi}\fi}

\newcommand{\structures}[1]{\ifnum#1=
88{{\sl Proceedings of the Third Annual Conference on
Structure in Complexity Theory\/}, IEEE, 1988}\else{\ifnum#1=
89{{\sl Proceedings of the Fourth Annual Conference on
Structure in Complexity Theory\/}, IEEE, 1989}\else{\ifnum#1=
90{{\sl Proceedings of the Fifth Annual Conference on
Structure in Complexity Theory\/}, IEEE, 1990}\else{\ifnum#1=
91{{\sl Proceedings of the Sixth Annual Conference on
Structure in Complexity Theory\/}, IEEE, 1991}\else{\ifnum#1=
92{{\sl Proceedings of the Seventh Annual Conference on
Structure in Complexity Theory\/}, IEEE, 1992}\else{\ifnum#1=
93{{\sl Proceedings of the Eighth Annual Conference on
Structure in Complexity Theory\/}, IEEE, 1993}\else{\ifnum#1=
94{{\sl Proceedings of the Ninth Annual Conference on
Structure in Complexity Theory\/}, IEEE, 1994}\else
This Structures not yet defined!
\fi}\fi}\fi}\fi}\fi}\fi}\fi}

\newcommand{\istcs}[1]{\ifnum#1=
93{{\sl Proceedings of the Second Israel Symposium on Theory
and Computing Systems\/}, 1993}\else{\ifnum#1=
95{{\sl Proceedings of the Third Israel Symposium on Theory
and Computing Systems\/}, 1995}\else
This ISTCS not yet defined!
\fi}\fi}

\def\|#1{\left| #1 \right|}

\def\boxfig#1
{\begin{figure}
\begin{center}
\fbox{
	\small
	\hbox{\quad
	\begin{minipage}{4.8in}
	#1
	\end{minipage}
	\quad}
	}
\end{center}
\end{figure}
}

\newcommand{\case}[2]
{
\begin{description}
\item[{\sf \underline{Case:} #1: }]
#2
\end{description}
}

\newtheorem{theorem}{Theorem}
\newtheorem{lemma}[theorem]{Lemma}
\newtheorem{fact}[theorem]{Fact}

\newtheorem{claim}[theorem]{Claim}

\newtheorem{definition}[theorem]{Definition}

\newcount\proofqeded
\newcount\proofended
\def\qed{ \mbox{\ \vrule width0.6ex height1em depth0cm}
\global\advance\proofqeded by 1 }
\newenvironment{proof}%
 {\proofstart}%
  {\ifnum\proofqeded=\proofended\qed\fi \global\advance\proofended by 1
  \medskip}
\makeatletter
\def\proofstart{\@ifnextchar[{\@oprf}{\@nprf}}
\def\@oprf[#1]{\paragraph{Proof of #1:}}
\def\@nprf{\paragraph{Proof:}}
\makeatother

\newcommand{\remove}[1]{}

\newcounter{problemct}
\input epsf

\def\calT{{\cal T}}

\newcommand{\ra}{\rightarrow}

\newcommand{\np}{\mbox{NP}}

\newcommand{\plurality}{\mbox{\rm plurality}}

\newcommand{\pfd}{P^{(f,d)}}
\newcommand{\pgd}{P^{(g,d)}}
\newcommand{\integers}{{\cal Z}^+}
\newcommand{\reals}{{\cal R}^+}

\ifnum\istcsabs=0

\makeatletter

\def\appearsin#1{\gdef\@appearsin{#1}}

\def\maketitle{\par
 \begingroup
 \def\thefootnote{\fnsymbol{footnote}}
 \def\@makefnmark{\hbox
 to 0pt{$^{\@thefnmark}$\hss}}
 \if@twocolumn
 \twocolumn[\@maketitle]
 \else \newpage
 \global\@topnum\z@ \@maketitle \fi\thispagestyle{plain}\@thanks
 \endgroup
 \setcounter{footnote}{0}
 \let\maketitle\relax
 \let\@maketitle\relax
 \gdef\@thanks{}\gdef\@author{}\gdef\@title{}\gdef\@appearsin{}
          \let\thanks\relax}
\def\@maketitle{\newpage
 \noindent \@appearsin
 \vskip 1in \begin{center}
 {\LARGE \@title \par} \vskip 1.5em {\large \lineskip .5em
\begin{tabular}[t]{c}\@author
 \end{tabular}\par}
 \vskip 1em {\normalsize \@date} \end{center}
 \par
 \vskip 1.5em}
\def\abstract{\if@twocolumn
\section*{Abstract}
\else \small
\begin{center}
{\bf Abstract\vspace{-.5em}\vspace{0pt}}
\end{center}
\quotation
\fi}
\def\endabstract{\if@twocolumn\else\endquotation\fi}
\mark{{}{}}

\fi

\newcommand{\polylog}{\mbox{\rm polylog}}

\newcommand{\corr}{{\sf Corr}}
\newcommand{\row}{{\sf row}}
\newcommand{\col}{{\sf col}}

\begin{document}

\title{{\LARGE\bf Some Improvements to Total Degree Tests}
\thanks{A version of this paper appeared in {\em Proceedings of
the 3rd Israel Symposium on Theory of Computing and Systems},
Tel Aviv, Israel,
January 4-7, 1995. This version corrects a few typographical
errors.}
\protect\vspace{10pt}}

\author{
{\sc Katalin Friedl}\thanks{\
Computer and Automation Research Institute, Hungarian
Academy of Sciences.
e-mail: {\tt kati@ilab.\allowbreak sztaki.hu.}.
Most of this work was done while at the Department of Computer
Science,
University of Chicago.  Supported in part by OTKA Grant 2581.} \and
{\sc Madhu Sudan}\thanks{\
Research Division,
IBM T.J.\ Watson Research Center, P.O.~Box 218,
Yorktown Heights, NY~10598, USA.
e-mail: {\tt madhu@watson.\allowbreak ibm.com}.
\protect\vspace{1in}}\protect\vspace{10pt}}

\date{}

\maketitle
\thispagestyle{empty}
\begin{abstract}

A low-degree test is a collection of simple, local rules
for checking
the proximity of an arbitrary function to a low-degree polynomial.  Each
rule depends on the function's values at a small number of places.
If a function satisfies many rules then it is close to a low-degree
polynomial. Low-degree tests play an important role in the development of
probabilistically checkable proofs.

In this paper we present two improvements to the efficiency of low-degree
tests. Our first improvement concerns the smallest field size over which
a low-degree test can work. We show how to test that a function is a
degree $d$ polynomial over prime fields of size only $d+2$.

Our second
improvement shows a better efficiency of the low-degree test of
\cite{RS} than previously known. We show concrete
applications of this improvement via the notion of ``locally checkable
codes''. This improvement translates into better tradeoffs
on the size versus probe complexity of probabilistically checkable proofs
than previously known.

\end{abstract}

\section{Introduction}

In this paper we consider functions mapping $m$ variables from some
finite field $F$ to the same field. Let the distance between two
functions $f$ and $g$, denoted by $d(f,g)$ be $\Pr_{x \in F^m} [ f(x) \ne
g(x)]$.
We use $\deg(f)$ to denote the {\em total
degree} of $f$, and $\deg_{\rm max} (f)$ to denote the largest
individual degree in any of the variables in $f$.

The low-degree testing problem for total degree (maximum
degree) is defined as follows:
\begin{definition}
For parameters $d \in \integers$ and $\delta,\epsilon \in \reals$, a
{\sf low-degree tester} is probabilistic oracle machine
$\calT$, that has access to a function $f:F^m \to F$ as an
oracle, and behaves as follows:
\begin{itemize}
\item If $\deg(f) \leq d$ ($\deg_{\rm max}(f) \leq d$) then $T^f$ accepts.
\item If for all total degree (maximum degree) $d$ polynomials $g$, $d(f,g) >
\epsilon$, then $T$ rejects with probability $\delta$.
\end{itemize}
\end{definition}
The low-degree testing problem has been studied widely due to their
relationship with probabilistically checkable (holographic) proofs and
program checking.
\cite{BFL,BFLS,FGLSS,AS,FHS,PS} study the case of testing the maximum
degree and \cite{BLR,GLRSW,RS,ALMSS} study the case of testing total
degree\footnote{Some of the improvements in the former family
have also affected the latter. In particular \cite{ALMSS}
obtain their improvement using the improved analysis of
\cite{AS}. Similarly the work of \cite{PS} also affects the
latter family of testers.}.
Our improvements are to the latter family of testers.
We start by describing their testers.
\begin{definition}
For points $x,h \in F^m$,
the {\em line} through $x$
with offset $h$
is the set of points $l_{x,h} = \{l_{x,h} (t) = x
+ t \cdot h | t \in F\}$.
\end{definition}
\begin{definition}
Given a function $f: F^m \to F$, a positive integer $d$
and points $x,h \in F^m$,
the line polynomial $\pfd_{x,h} : F \to F$ is a univariate
polynomial of degree at most $d$ which satisfies
$\pfd_{x,h}(t) = f(x + t \cdot h)$ for the most $t$. {\em
Ties are broken arbitrarily.}
\end{definition}
The tester in \cite{RS} is effectively the following: ``Pick $x$,
$h$ randomly and uniformly from $F^m$ and $t$ randomly from
$F$ and verify that $\pfd_{x,h}(t) = f(x + t \cdot h)$.''
For our purposes it is not important how $\pfd$ is computed by the
tester. This will become clear in the context of our applications.
\remove{
Notice that if $f$ is a
polynomial then $\pfd_{x,h}$ can be easily interpolated by
sampling $f$ on $O(d)$ points of the form
$\{x + t_i h\}_{t_i \in F}$. In the case that $f$ is not close
to a polynomial then this function does not remain easy to
compute. Here, they continue to use a polynomial interpolated
from the value of $f$ on the set $\{x + t_i h\}_{t_i \in F}$.
and use the observation that for a given
$x,h$ pair, of all the degree $d$ polynomials $P_{x,h}:F \to F$,
the polynomial $\pfd_{x,h}$ is the one most likely to pass the
test $\pfd_{x,h}(t) = f(x + t \cdot h)$. Thus }
The
correctness of the tester is proved by the following kind of a statement.

{\bf Informal Statement:} {\em
If $|F|$ is a sufficiently large function of $d$ and $\delta$ is a
sufficiently small function of $d$, then
given a function $f:F^m \to F$, if there exists a set of degree
$d$ polynomials $\{P_{x,h}\}$ which satisfies
$$
\Pr_{x,h,t} \left[ P_{x,h} (t) \ne f(x + t\cdot h) \right]
\leq \delta
$$
then there exists a degree $d$ polynomial $g:F^m \to F$ such
that $d(f,g) \leq 2\delta$.
}

The above statement does not specify the conditions on $|F|$ and
$\delta$, and determining the exact conditions on these
parameters turn out to be the interesting aspect in the
analysis of low-degree tests. The improvements noted in this
paper apply to these two parameters.

\subsection{Reducing Field Size}

The motivation for the following theorem is primarily one of
curiosity. The smallest field size over which polynomials of a
given total degree exhibit sufficient ``redundancy'' to, say,
enable the application of the Schwartz-Zippel like theorems
\cite{Schwartz,Zippel}, is when the field size is at least
$d+2$. The low-degree tester of \cite{GLRSW} uses
sets of the same size, i.e., $d+2$,
as elementary test sets.
Their proof
manages to show that in a certain sense (see Lemma~\ref{robust})
fields of size $d+2$ are sufficient to show some sort of
robustness. However their proof falls short of showing
low-degree tests that work over fields of size $d+2$ because
of the  lack of an ``exact characterization'' (in the sense of
\cite{RS:jv}). We complement their work by providing an
exact characterization of low-degree polynomials, which shows
that their tester is good for prime fields of size $d+2$, and
improves the bound for non-prime fields as well. We give
examples to show that our characterizations are essentially
the best possible.
\begin{lemma}
If $q$, the order of $F$, and $p$, its characteristic, satisfy $q - q/p -
1 \geq d$ and $g : F^m \to F$ satisfy$$\pgd_{x,h} (t) = g(x + t
\cdot h) \mbox{ for all } x,h,t $$
then $g$ is a degree $d$ polynomial.
\end{lemma}
We use the above statement in combination with the following
Lemma from \cite{GLRSW,RS:jv}, to get
Theorem~\ref{d+2theorem}.
\begin{lemma}[\cite{GLRSW,RS:jv}]
\label{robust}
There exists a constant $c$ such that if $|F| \geq d+2$ and $\delta
\leq {1 \over c(d+1)^2}$ and $P$ and $f$ satisfy
$$
\Pr_{x,h,t} \left[
P_{x,h} (t) \ne f(x + t\cdot h) \right]
\leq \delta
$$
then there exists a function $g: F^m \to F$ such that $d(f,g) \leq 2
\delta$ and $g$ satisfies
$$\pgd_{x,h} (t) = g(x + t \cdot h) \mbox{ for all } x,h
\in F^m \mbox{ and } t \in F \mbox{.}
$$
\end{lemma}
\begin{theorem}
\label{d+2theorem}
Let $q$ denote the order of the field $F$ and $p$ its
characteristic. Then there exists a constant $c$ such
that if $q - q/p - 1 \geq d$ and $\delta \leq {1 \over
c(d+1)^2}$ and $P$ and $f$ satisfy
$$
\Pr_{x,h,t} \left[ P_{x,h} (t) \ne f(x + t\cdot h) \right] \leq \delta
$$
then there exists a degree $d$ polynomial $g : F^m \to F$ such
that $d(f,g) \leq 2\delta$.
\end{theorem}
In Section~\ref{exact} we also show that the requirement on
$|F|$ is the tightest possible in the following sense: For all
$p,q,d$ such that $d > q - q/p -1$, we show that there exist
functions $g : F^m \to F$ and $P_{x,h} : F \to F$, such that
$\deg(P_{x,h}) \leq d$ and for all $x,h \in F^m, t \in F$,
$P_{x,h}(t) = g(x + t\cdot h)$, but $\deg(g) > d$.

\subsection{Improving the efficiency}

Improving the second of the two parameters in the statement of the
Informal Statement is a task of greater significance. (Here an
improvement would imply a larger value of $\delta$.)
The result in \cite{GLRSW} shows that the test works for $\delta \leq
O(1/d^2)$. The improvements in \cite{RS} and \cite{ALMSS} yielded
$\delta \leq O(1/d)$ and $\delta \leq \delta_0$ for some
$\delta_0 > 0$ respectively. The constant $\delta_0$ coming from the
latter analysis is
not described explicitly but the number appears to be
fairly small. Here we show that the theorem works for any $\delta < 1/8$.
More precisely,
\begin{theorem}
\label{What}
For every $\epsilon > 0$, there exist
$c < \infty$ such that for all $d$, if $|F| \geq cd$ the following holds.
Given a function $f:F^m \to F$ and degree $d$ polynomials
$\{P_{x,h}\}$ satisfying
$$ \Pr_{x,h,t} \left[ P_{x,h} (t) \ne f(x + t \cdot h) \right]  =
\delta \leq 1/8 - \epsilon$$
there exists a degree $d$ polynomial $g$ such that $d(f,g) \leq
2\delta$.
\end{theorem}
{\bf Remark:}
The bound on the field size in the above theorem is also
better than that of \cite{ALMSS} who are only able to show it
for $|F| \geq d^3$. However, this improvement can already be
inferred in the work of \cite{PS}. In fact, our analysis inherits
this particular improvement from their analysis.

\subsection{An Application}

The second theorem given above turns out to have some implication to
the properties of probabilistically checkable proofs. In particular
it implies that the proofs constructed in \cite{ALMSS} have a much
smaller probe complexity than shown earlier. It turns out that
all the known construction of holographic proof inherits part of
their properties from the properties of some underlying
error-correcting code.
It is easiest to describe the effect of our analysis in terms
of the improvement in the properties of the codes created in
\cite{ALMSS,RS:jv}.
The following definitions are from \cite{RS:jv}.

An {\sf $(k,n,d,a)$-code} consists of an alphabet $\Sigma$ such that $
\log |\Sigma| = a$ and a function $C : \Sigma^k \ra
\Sigma^n$, such that for any two strings $m, m' \in \Sigma^k$, the
(Hamming) distance between $C(m)$ and $C(m')$ is at least $d$.
For applications to probabilistically checkable proofs, it is
sufficient to consider codes restricted a small range
of these parameters. We call these the
{\em good} codes. Such codes need to have constant relative distance.
The encoded message is allowed to be much larger than the original
message size, as long as the final length is polynomially bounded.
\begin{definition}
[Good Code]
A family of codes $\{C_i\}$ with parameters $(k_i,n_i,d_i,a_i)$
is good if $k_i \rightarrow \infty$,
$n_i$ is upper bounded by some polynomial in $k_i$,
$d_i/n_i > 0$ and $a_i = \polylog(k_i)$.
\end{definition}
For the application to probabilistically checkable proofs, the
primary question is the following: ``Does the code admit {\em
very simple} randomized error detection?''.
This notion is formalized in the next definition. Intuitively
the definition says that the error-detection can be done by
probing just $p$ letters of a word to get a confidence
$\delta$ that it is close to some codeword.
\begin{definition}
For a positive integer $p$
and a positive real number $\delta$,
an $(n,k,d,a)$-code $C$ over the alphabet $\Sigma$
is {\bf $(p,\delta)$-locally
testable} if the following exist
\begin{itemize}
\item A probability space $\Omega$ which can be {\em efficiently}
sampled.
\item Functions $q_1, q_2, \ldots, q_p : \Omega \ra
\{1,\ldots,n\}$.
\item A boolean function $V : \Omega \times \Sigma^p \ra \{0,1\}$.
\end{itemize}
with the property that for all $w \in \Sigma^n$, if
$$\Pr_{r \in \Omega} \left [
V ( r, w_{q_1(r)}, \ldots, w_{q_p(r)} ) = 0
\right ] <  \delta $$
then there exists a (unique) string $m \in \Sigma^k$ such that
$d(w, C(m)) < d/2$.
Conversely, if $w = C(m)$ for some $m$, then $V ( r, w_{q_1(r)},
\ldots,
w_{q_p(r)} ) = 1$ for all $r \in \Omega$.
\end{definition}
The codes of \cite{BFLS} for instance produce good codes which are
$(\polylog(k_i), \Omega(1))$-locally testable.
The work of \cite{AS} implicitly describe a related code which
achieves both $p,{1 \over \delta} = O(1)$, but requires very
large alphabet sizes to get this -- namely their code requires
$a_i = k_i^{\epsilon}$.
The significant improvement in \cite{ALMSS}
is to get good codes which have $p=2$, $\delta > 0$ with $a_i =
\polylog(k_i)$. (By applying a recursive technique introduced
by \cite{AS} to this code they later manage to reduce $a$ to a
constant as well.)
The code used by \cite{ALMSS} is the following (see also
\cite{RS:jv}):
\begin{definition}
[Polynomial-Line Codes]
\label{defn:poly-line}
Let $c_1 > 1$ and $c_2 \geq 1$ be parameters.
The {\em polynomial-line codes}
$\{L_m\}$ are chosen
by letting $d = \Theta(m^{c_1})$ and picking a
finite field $F$ of size $\Theta(d^{c_2})$.
The code works over the alphabet $\Sigma = F^{d+1}$.
The message consists of
${m+d \choose d}$ field elements (or ${m+d \choose d}/(d+1)$ letters from
$\Sigma$) and is viewed as an $m$-variate  degree $d$ polynomial
specified by its coefficients.
Given a message polynomial $f$, the codeword
consists of $\{\pfd_{x,h}\}_{x,h \in F^m}$ where $\pfd_{x,h}$
is the line polynomial for the line $l_{x,h}$ described by its
$d+1$ coefficients.
The code achieves
$k_m = {m +d \choose d}/(d+1)$ and $n_m = |F|^{O(m^2)}$ over the
alphabet $F^{d+1}$.
\end{definition}
It is clear that for all constants $c_1$ and $c_2$ the Polynomial-Line
codes are good codes. \cite{ALMSS} show that for all $c_1 > 1$
and $c_2 \geq 3$ these codes give $(2,\delta>0)$-locally
testable codes. \cite{PS} improve this to $c_2 \geq 1$,
without changing the $\delta$ in any significant way. Our
analysis (Theorem~\ref{What})
immediately yields that Polynomial-Line Codes are
$(2, 1/8 - \epsilon)$-locally testable. It can
be easily shown that no code can achieve $(2, 1/2 + \epsilon)$-local
testability. Thus in this case
our results come close to optimality.

\paragraph{Connection with proof checking}

Lastly we describe a very informal manner the way
in which this affects the
construction of probabilistically checkable proofs. We
assume that the reader of this subsection is familiar with the
notion of probabilistically checkable proofs (PCPs) as defined
in \cite{AS} (see, for instance, \cite{Babai:survey} for a
survey). In particular we discuss the probe complexity of
proofs and the sizes of probabilistically checkable proofs.

As mentioned earlier every holographic proof ends up inheriting
part of its properties on some underlying locally-testable
code. In order to test that a given proof is valid one ends up
testing that the proof corresponds to a valid codeword. This
effectively implies that to obtain a fixed degree of
confidence, one has to look at $O(p/\delta)$ letters in the
proof. Thus the probe complexity of a PCP seems to be
inherently dependent on the ratio of $p$ and $\delta$.

However, the relationship between $p/\delta$ and the probe
complexity of PCP turns out to be not so simple. \cite{BGLR}
manage to reduce the probe complexity of a PCP
to about 24 bits to get a confidence of $1/2$ (from some
unknown number estimated to be around $10^4$ in \cite{ALMSS})
without improving the analysis of low-degree tests! How do
they obtain this reduction? It turns out that this reduction
is obtained by exploding the proof size to the order of
$n^{10^4}$ (from some smaller polynomial of size about
$n^{12}$ in
\cite{ALMSS}). But by incorporating the analysis from
this paper into the
analysis of PCP one can obtain better bounds on the probe
complexity of proof systems. The verifier we construct
probes a proof at most 165 bits (as opposed to the $10^4$ of
\cite{ALMSS,PS}) while increasing the proof size to only
$n^{2+\epsilon}$ (to be contrasted with the $n^{10^4}$ in \cite{BGLR}).
(We point out that the improvement relies fairly heavily on
the techniques developed in \cite{BGLR} and \cite{PS}, as well as
those of \cite{AS} and \cite{ALMSS}.)

\remove{
\section{Characterizing the Total Degree of Polynomials}  
\label{exact}

Let $F = F_q$ be a finite field of order $q = p^s$ where $p$ is its
characteristic.

\begin{theorem}\label{totaldegree}
Let $g: F^m \to F$ be a function which satisfies
$$\forall x,h \in F^m, t \in F ~~~~ \pgd_{x,h}(t) = g(x + t \cdot h).$$
Then if $q - q/p - 1 \geq d$, then g is a polynomial of degree at most
$d$.
\end{theorem}

{\bf Remark.}  The inequality $q-q/p-1 \ge d$ in the Theorem
cannot be weakened for any $q$. Indeed, consider
the 2-variable function
$g(x_1,x_2) =  (x_1^{(p-1)} x_2)^{q/p}$.
For every pair $x,h$, consider the formal expansion of $\pgd_{x,h}$.
Each term in this univariate polynomial has degree
$\le (p+1)q/p = q + q/p$ and each exponent is divisible by $q/p$.
As a function we have $t^q = t$ , hence
the degree of $\pgd_{x,h}$ is at most $q - q/p$.
On the other hand the total degree of $g$ is $q > (q-q/p)$.

For the proof of Theorem~\ref{totaldegree}, first we note
two easy facts about polynomials over $F_q$:    

\begin{lemma} \label{modp-1}
If the value of the formal polynomial $h(t) = \sum_{i=0}^m c_i t^i$
is equal to a polynomial of degree at most $d < q-1$ for every $t \in F_q$
then $\sum_{s=0}^{\infty} c_{d+1+(q-1) s} = 0$.
\end{lemma}

\begin{proof}
Use the identity $t^q = t$ (over $F_q$).
\end{proof}

\begin{lemma} \label{zero}
If a formal polynomial $h(x_1,\ldots,x_m)$ has $\deg_{max} \le q-1$
and it vanishes over $F_q^m$ then $h$ is formally zero.
\end{lemma}

\begin{proof}
Follows by induction on $m$.
\end{proof}

We also need a lemma about the behavior of the binomial coefficients
modulo $p$.

\begin{lemma} \label{binom}
Let $0 < r \le n \le p^s -1$. If  $r = k p^{s-1}$
then ${n \choose r}$ is not  divisible by $p$.
\end{lemma}

\begin{proof}
For any positive integer $l$, the largest power of $p$ that
divides $l!$ is $\lfloor l/p \rfloor + \lfloor l/p^2 \rfloor +
\lfloor l/p^3\rfloor + \cdots$.
But for $r = k p^{s-1}$, the identity
$\lfloor n/p^i \rfloor = \lfloor n/p^i \rfloor +
\lfloor (n-r)/p^i\rfloor$ holds.
Therefore $n!$ and $r! (n-r)!$ are divisible by exactly the
same power of $p$.
\end{proof}

\begin{proof}[Theorem~\ref{totaldegree}]
Over $F_q$, every function is represented by a formal polynomial of
$\deg_{max} \le q-1$,  and by Lemma~\ref{zero}, this polynomial is unique.
Let us regard $g$ as such a polynomial.

Applying the condition to axis parallel lines, it follows
(using Lemma~\ref{zero}) that the degree of $x_i$ in $g$ is $\le d$.
This means that $g$ can be written in the form
$$g(x_1, \ldots, x_m) =
\sum_{i_j \le d} \alpha_{i_1, \ldots, i_m} x_1^{i_1} \cdots
x_m^{i_m}.$$

We show by induction on $m$ that $\alpha_{i_1, \ldots, i_m} = 0$
whenever $\sum_{j=1}^m i_j \ge d+1$.

By the inductive hypothesis
\begin{equation} \label{induction}
\sum_{j=1}^{m-1} i_j \ge d+1 \quad \Longrightarrow
\quad \alpha_{i_1, \ldots, i_m} = 0.
\end{equation}
This is so because if we collect the terms
$x_1^{i_1} \cdots x_m^{i_m}$, where $i_1, \ldots, i_{m-1}$ is fixed
and $i_m$ varies, we obtain $x_1^{i_1} \cdots x_{m-1}^{i_{m-1}}$ times
a polynomial $h_{i_1, \ldots, i_{m-1}}(x_m)$ of degree $\le d$.
Any substitution $x_m = b$ in $f$ gives a polynomial of $(m-1)$-variables
which by the inductive hypothesis has total degree $\le d$.
Therefore for any $b \in F_q$ and any $i_j$ such that
$\sum_{j=1}^{m-1} i_j \ge d+1$, $h_{i_1, \ldots, i_{m-1}}(b) = 0$.
This implies that for these $i_j$ values the polynomial
$h_{i_1, \ldots, i_{m-1}}$  is formally zero.
The coefficient of $x_m^{i_m}$ in $h_{i_1, \ldots, i_{m-1}}$ is exactly
$\alpha_{i_1, \ldots, i_m}$. This completes the proof of
(\ref{induction}).

Consider now the polynomial $h(t)$ obtained by the following formal
substitutions: $x_i=a_i t$, $1 \le i \le m-1$ and $x_m = t+b$,
where $a_i, b \in F_q$.
For fixed $a_i$ and $b$, $g$ is a polynomial in $t$.
By (\ref{induction}) the largest exponent $\ell$ of $t$ that appears in
$h$ is
$\ell \le 2d \le d + q- q/p -1 \le d + q-2$.
By assumption $h$ agrees with a polynomial of degree $\le d$
and so Lemma~\ref{modp-1}
implies that  the coefficient of $t^{\ell}$ for
$d+1 \le \ell \le q-1$ has to be zero
(for every choice of $a_i, b$).
The coefficient of $t^{\ell}$ is
$$\sum_{\stackrel{i_j \le d}{r + \ell = \sum i_j}}
      \alpha_{i_1, \ldots i_m}\, a_1^{i_1} \cdots a_{m-1}^{i_{m-1}}\,
      {i_m \choose r}\, b^r.$$
This is a polynomial in $b$ of degree $\max(\sum i_j) - \ell \le
2d-(d+1)$,
therefore the coefficient of $b^r$ has to be zero for all choices
of $a_1, \ldots, a_{m-1}$. This coefficient is a polynomial of the
$a_i$ of $\deg_{max} \le d$; therefore by Lemma~\ref{zero}
it is formally zero, which means that
$\alpha_{i_1, \ldots i_m}\, {i_m \choose r} = 0$
for all values of $i_j \le d$ and $r = \sum i_j - \ell$
($d+1 \le \ell \le q-1$).
For fixed $i_j$ this gives $q-d-1$
possible values for $r$. By assumption $q - d- 1 \ge q/p = p^{s-1}$,
there
is an $r$ in the form $k p^{s-1}$ in the interval.
For this $r$, ${i_m \choose r}$ is not zero $\bmod\,p$
(Lemma~\ref{binom}),
which shows that $\alpha_{i_1, \ldots i_m} = 0$ holds.
\end{proof}
}

\section{Characterizing the Total Degree of Polynomials}  
\label{exact}

Let $F = F_q$ be a finite field of order $q = p^s$ where $p$ is its
characteristic.
\begin{theorem}\label{totaldegree}
Let $g: F^m \to F$ be a function which satisfies
$$\forall x,h \in F^m, t \in F ~~~~ \pgd_{x,h}(t) = g(x + t \cdot h).$$
Then if $q - q/p - 1 \geq d$, then g is a polynomial of degree at most
$d$.
\end{theorem}
{\bf Remark.}  The inequality $q-q/p-1 \ge d$ in the Theorem
cannot be weakened for any $q$. Indeed, for any $d$ such that
$q - q/p - 1 < d < q$,
consider the bivariate function
$g(x_1,x_2) =  (x_1^{(p-1)} x_2)^{q/p}$.
For every pair $x,h$, the univariate function $P_{x,h}$
given by $P_{x,h}(t) = g(x + t\cdot h)$. Each term in this
univariate polynomial has degree at most
$(p+1)q/p = q + q/p$ and each exponent is divisible by $q/p$.
As a function we have $t^q = t$ , and thus
$\deg(P_{x,h})$ is at most $q - q/p$.
Thus we have $\pgd_{x,h} \equiv P_{x,h}$.
On the other hand the total degree of $g$ is $q > d$.

For the proof of Theorem~\ref{totaldegree},
we first prove a lemma about the behavior of the binomial coefficients
modulo $p$.
\begin{lemma} \label{binom}
Let $0 < r \le n \le p^s -1$. If  $r = k p^{s-1}$
then ${n \choose r}$ is not  divisible by $p$.
\end{lemma}
\begin{proof}
For any positive integer $l$, the largest power of $p$ that
divides $l!$ is $\lfloor l/p \rfloor + \lfloor l/p^2 \rfloor +
\lfloor l/p^3\rfloor + \cdots$.
But for $r = k p^{s-1}$, the identity
$\lfloor n/p^i \rfloor = \lfloor r/p^i \rfloor +
\lfloor (n-r)/p^i\rfloor$ holds.
Thus the largest power of $p$ that divides $n!$ is $\sum_{i=1}^{\infty}
\lfloor n/p^i \rfloor =    \sum_{i=1}^{\infty} (\lfloor r/p^i \rfloor +
\lfloor (n-r)/p^i\rfloor)$.
Therefore $n!$ and $r! (n-r)!$ are divisible by exactly the
same power of $p$.
\end{proof}
\begin{proof}[Theorem~\ref{totaldegree}]
Assume for the sake of contradiction that the assertion of the theorem is false.
Let $m$ be the smallest positive integer for which the following holds:
\begin{eqnarray}
\lefteqn{\exists g : F^m \to F \mbox{s.t. }  \forall x,h \in F^m, t \in F }
\nonumber \\
& & \pgd_{x,h}(t) = g(x + t\cdot h) \mbox{ but } \deg(g) > d.
\label{eqn:contr}
\end{eqnarray}
Express $g$ in the form:
$$g(x_1,\ldots,x_m) = \sum_{i_1=0}^{q-1}\cdots\sum_{i_m=0}^{q-1}
\alpha_{i_1,\ldots,i_m} x_1^{i_1} \cdots x_m^{i_m}.$$
(Notice that there exists $\alpha$'s such that the above is true, and these are
unique.) Since $g$ is not a degree $d$ polynomial, there exist $l$ and
$i_1,\ldots,i_m$
such that $\sum_{j=1}^m i_j = l > d$ and $\alpha_{i_1,\ldots,i_m} \ne 0$. Let
$l$ be the largest integer with this property. We consider the following cases:
\case{$\sum_{j=1}^{m-1} i_j > d$}
	{
	We show that this contradicts the assumption that $m$ is the smallest
	integer for which (\ref{eqn:contr}) holds. For $a_m \in F$, let
	$g_{a_m}:F^m \to F$ be given by
	$g_{a_m}(x_1,\ldots,x_{m-1}) = g(x_1,\ldots,x_{m-1}, a_m)$. Notice first
	that $g_{a_m}$ satisfies $g_{a_m}(x + t\cdot h) = P^{(g_{a_m},d)}_{x,h}
	(t)$ for all $x,h \in F^{m-1}$ and $t \in F$. This follows from the fact
	that $P^{(g_{a_m},d)}_{x,h} (\cdot) = \pgd_{x',h'}(\cdot)$ for $x' =
	<x,a_m>$ and $h' = <h,0>$. We now show that there exists $a_m$ such that
	$\deg(g_{a_m}) > d$. Observe that the coefficient for $x_1^{i_1}\cdots
	x_{m-1}^{i_{m-1}}$ is $\sum_{i=0}^{q-1} \alpha_{i_1,\ldots,i_{m-1},i}
	a_m^i$. This summation is a non-zero polynomial in $a_m$ of
	degree less than $q$.
	Thus there must exist a point $a_m$ where the summation
	is non-zero. This gives us $a_m$
	such that $g_{a_m}$, a function of $m-1$ variables, satisfies (\ref{eqn:contr}).
	As promised, this violated the minimality of $m$.
	}
\case{$\sum_{j=2}^{m} i_j > d$}{Similar to above.}
\case{$l = \sum_{j=1}^m i_j < q$}
	{
	For $a_1,\ldots,a_m \in F$, let	$g_{a_1,\ldots,a_m}: F \to F$ be given
	by $g_{a_1,\ldots,a_m}(t) = g(a_1t,\ldots,a_mt)$. The coefficient of
	$t^l$ in $g_{a_1,\ldots,a_m}$ is given by
	
$$\sum_{
k_1,\ldots,k_m \\
\mbox{ s.t. } k_1 + \ldots + k_m = l
}
	\alpha_{k_1,\ldots,k_m} a_1^{k_1} \cdots a_m^{k_m}.$$

	Since this expression is a polynomial in the $a_j$'s of degree less
	than $q$ and is not identically zero, there exist
	$a_1,\ldots,a_m$ for which the coefficient of $t^l$ in
	$g_{a_1,\ldots,a_m}(t)$ is non-zero. But for $x = 0$ and $a =
	<a_1,\ldots,a_m>$, we find that $\pgd_{x,a} \equiv g_{a_1,\ldots,a_m}$.
	and the fact that $g_{a_1,\ldots,a_m}$ is not a polynomial of degree $d$
	contradicts the conditions guaranteed in
	(\ref{eqn:contr}).
	}
\case{None of the above}
	{
	In this case we have $\sum_{j=1}^{m-1} i_j \leq d$, $i_m \leq \sum_{j=2}^{m}
	i_j \leq d$ and $l \geq q$.
	Here we consider the function
	$g_{a_1,\ldots,a_{m-1},b}(t) = g(a_1t,\ldots,a_{m-1}t, b + t)$
	and show that for some choice of $a_1,\ldots,a_{m-1}$ and $b$, the
	coefficient of $t^r$ in $g_{a_1,\ldots,a_{m-1},b}$ is non-zero, for
	some $r$ in the range $[d+1,q]$, of the form
	$n p^{s-1} + \sum_{j=1}^{m-1} i_j$.
	Such a choice for $r$ exists since the range $[d + 1,q]$
	contains at least $q - d- 1 \geq q/p = p^{s-1}$ elements.

	We start with the observation that the coefficient of $t^r$ in the
	function $g_{a_1,\ldots,a_{m-1},b}(t)$ is the same as the coefficient of
	$t^r$ in the formal power series expansion of $g$ with the formal
	substitutions $x_i =
	a_i t$ and $x_m = b + t$. This is true because the formal power series
	contains terms of degree at most $l$ and $l$ satisfies the condition:
	$q + r > l$. (Since $l = i_m + \sum_{j=1}^{m-1} i_j \leq i_m + r
	\leq d + r < r + q$.)

	The coefficient of $t^r$ in the formal expansion is
{\small
$$\sum
{k_m \choose r - \sum_{i=1}^{m-1} k_i}
\alpha_{k_1,\ldots,k_m}
\prod_{i=1}^{m-1}a_i^{k_i}
b^{\sum_{i=1}^m k_i - r}
$$}
where the summation ranges over all choices of $k_1,\ldots,k_m$ such that
$r \leq \sum_{j=1}^m k_j \leq l$.
\remove{
$$\sum_{
\begin{array}{c}
k_1,\ldots,k_m \\
\mbox{ s.t. } r \leq \sum_j k_j \leq l
\end{array}
}
\left(
\begin{array}{lc}
&{k_m \choose r - k_1 - \cdots - k_{m-1}} \\
\times & \alpha_{k_1,\ldots,k_m} \\
\times & a_1^{k_1} \cdots a_{m-1}^{k_{m-1}} \\
\times & b^{k_1 + \cdots + k_m - r}
\end{array}
\right).$$
}
	Thus coefficient is a polynomial in $a_j$'s and $b$ of degree at most
	$q$ in each variable. Moreover for $k_j = i_j$, the term
	$\alpha_{k_1,\ldots,k_m}$ is non-zero and the term
	${k_m \choose r - k_1 - \cdots - k_{m-1}}$ simplifies to ${i_m \choose
	n p^{s-1}}$ which is also non-zero (by Lemma~\ref{binom}).
	Thus the coefficient of $t^r$ is a non-zero polynomial of maximum
	degree at most $q$.
	Hence there exists a
	choice of $a_1,\ldots,a_{m-1}$ and $b$ such that the coefficient of
	$t^r$ is non-zero modulo $p$.

	We now obtain the contradiction in the usual way.
	We observe that $g_{a_1,\ldots,a_{m-1},b} \equiv \pgd_{x,h}$ for
	$x = <0,\ldots,0,b>$ and $h = <a_1,\ldots,a_{m-1},1>$. Thus
	$g_{a_1,\ldots,a_{m-1},b}$ should be a polynomial of degree at most $d$,
	contradicting the fact that the coefficient of $t^r$ is non-zero.
	}
\end{proof}

\section{Efficiency of the Lines test}

The main theorem of this section is motivated by the following
tester: The tester ${\cal T}$ is provided access to an oracle for
$f : F^m \to F$ and an augmenting oracle $O:F^{2m} \to F^{d+1}$.
The augmenting oracle takes as input the description of a line
by the pair $x,h \in F^m$ and provides the coefficients
of the ``line polynomial'' $P_{x,h}$. The effect of
Theorem~\ref{main theorem} is to show that the tester behaves as
follows:
\begin{itemize}
\item If $f$ is a degree $d$ polynomial then there exists an $O$
such that $\calT^{f,O}$ always accepts.
\item If $d(f,g) \geq 1/4$ for every degree $d$ polynomial $g$,
then for every oracle $O:F^{2m}\to F^{d+1}$, $\calT^{f,O}$
rejects with probability at least $1/8 - \epsilon$.
\item $\calT$ makes exactly one call to each oracle (i.e., $f$ and $O$).
\end{itemize}

The consequences of this theorem are summarized in
Section~\ref{next}.
\begin{theorem}
\label{main theorem}
For every $\epsilon > 0$, there exists
$c < \infty$ such that for all $d \in \integers$ if $|F| \geq cd$ then
the following holds.
Given a function $f:F^m \to F$ and degree $d$ polynomials
$\{P_{x,h}\}$ such that:
$$ \Pr_{x,h,t} \left[ P_{x,h} (t) \ne f(x + t \cdot h) \right]  =
\delta \leq 1/8 - \epsilon$$
there exists a degree $d$ polynomial $g$ such that $d(f,g) \leq
2\delta$.
\end{theorem}
Our proof is based on the proof in \cite{RS:jv} and borrows various
ingredients from their technique. However our analysis
seems to be simplify certain aspects of their proof by introducing an
inductive analysis to their proof. The improvement in the value of
$\delta$ is obtained by very careful sampling of the underlying space and
the application of pairwise independent analysis to their space. The use
of pairwise independent analysis in low-degree testing seems to be new.

In what follows  we fix an $\epsilon > 0$. We assume that $c \to
\infty$. Thus whenever the notation $\alpha = o(1)$ is used in what
follows, it implies that $\alpha \to 0$ as $c \to \infty$.

We start with a couple of definitions. Given a function $f: F^m \to F$,
let $\delta_f$ be defined as
$$\delta_f = \Pr_{x,h} \left[ f(x) \ne \pfd_{x,h} (0) \right]$$
and let $\corr_f : F^m \to F$ be the function defined by
$$\corr_f (x) = \plurality_h \{ P_{x,h}^{(f)} (0) \}
\footnote{The {\em plurality} of a multiset is the most commonly occurring
element in the multiset. We use the word plurality as opposed to majority
since the latter could also be used to point to the (unique) element that
occurs with frequency more than half.}.$$
We start with a few basic facts about $\delta_f$ and $\corr_f$.
\begin{fact}
\label{deltadeltaf}
For any function $f:F^m \to F$, and degree $d$ polynomials $\{P_{x,h}:F
\to F\}_{x,h \in F^m}$,
$$\Pr_{x,h,t} \left[ f(x + t\cdot h) \ne P_{x,h}(t) \right] \geq
\delta_f.$$
\end{fact}
The above fact follows directly from the fact that for each $x,h$
$P^{f,d}_{x,h}$ minimizes (over random  $t$) the probability that $f(x +
t\cdot h) \ne P_{x,h}(t)$.
\begin{lemma}
[\cite{GLRSW}]
\label{2deltaf}
$d(f,\corr_f) \leq 2 \delta_f$.
\end{lemma}
\begin{lemma}
[\cite{GLRSW}]
\label{1/4th}
For all $\beta > 0$, if $g$ is a degree $d$ polynomial such that
$d(f,g) < 1/4 - \beta$,
then $\corr_f \equiv g$.
\end{lemma}
We need a slightly stronger version of the above lemma for our purposes
which we prove next.
\begin{lemma}
\label{1/2th}
For all $\beta > 0$,
if $g:F^m \to F$ is a degree $d$ polynomial such that
$d(f,g) < 1/2 - \beta$
then $d(\corr_f,g) = o(1)$.
\end{lemma}
\begin{proof}
Consider randomly chosen $x,h \in F^m$ and the line $l_{x,h}$. Notice
that this line represents a pairwise independent collection of points
from $F^m$. Thus with probability $1 - \alpha$, where $\alpha
= o(1)$,
the number of points, $y$,
on $l_{x,h}$ such that $f(y) \ne g(y)$ is less than $1/2 - \epsilon/2$
and in such cases $\pfd_{x,h} \equiv \pgd_{x,h}$.

Now consider the set $B = \{x | \pfd_{x,h} \not\equiv \pgd_{x,h}
\mbox{ for a majority of } h \in F^m\}$. Based on the above argument
notice that the fraction $|B|/|F|^m$ is at most $2\alpha =
o(1)$. But for $x \not\in
B$, we have $\corr_f(x) = g(x)$.
\end{proof}

The main lemma we prove is the following:
\begin{lemma}
\label{main lemma}
$\forall f: F^m \to F \mbox{ s.t. } |F| > 16/\epsilon^2 \mbox{ and }
\delta_f \leq 1/8 - \epsilon$, $\delta_{\corr_f}
< \delta_f$.
\end{lemma}
We defer the proof to the next subsection. We first show why this
suffices.
\begin{proof}[Theorem~\ref{main theorem}]
We prove this theorem by induction of $\delta$. (Observe that since
we are talking of functions over finite domains, $\delta$ can only take
finitely many values.) Say the theorem is true for functions
$f,\{P_{x,h}\}$ with
$$\Pr_{x,h,t} \left[P_{x,h}(t) \ne f(x +
t\cdot h)\right] < \delta.$$
Now consider functions $f,\{P_{x,h}\}$ with
$$\Pr_{x,h,t} \left[P_{x,h}(t) \ne f(x +
t\cdot h)\right] = \delta.$$
For such a $f$ consider the function $\corr_f$. By Lemma~\ref{main
lemma},
$$ \Pr_{x,h,t} \left[ P^{(\corr_f,d)}_{x,h} (t) \ne \corr_f(x + t\cdot h)
\right] = \delta_{\corr_f} < \delta_f.$$
By induction there exists a degree
$d$ polynomial $g$ such that $d(\corr_f,g) \leq
2\delta_{\corr_f} < 2\delta_f$. By
Lemma~\ref{2deltaf} $d(f,\corr_f) \leq 2\delta_f$. Thus $d(f,g) \leq
4\delta_f \leq 4 \delta \leq 1/2 - 4\epsilon$.
By Lemma~\ref{1/2th} $d(\corr_f,g) =
o(1)$. This in turn implies that $d(f,g) \leq 1/4 - \epsilon + o(1)$. By
Lemma~\ref{1/4th} we now conclude that $\corr_f = g$ implying that
$d(f,g) \leq 2\delta_f$.
\end{proof}

\subsection{Proof of Main Lemma}

The proof of Lemma~\ref{main lemma} relies on a minor strengthening of the
following lemma due to \cite{PS}, which in turn improves upon a similar
lemma in \cite{AS}.
\begin{lemma}
[\cite{PS}]
\label{ASPS}
For any $\epsilon > 0$, if
$r_i$ and $c_j$ are families of degree $d$ polynomials such that
$$\Pr_{i,j \in F} \left[ r_i(j) \ne c_j(i) \right] \leq 1/4 - \epsilon,$$
then there exists a bivariate polynomial $Q$ of degree $d$ in each variable
such that
$$\Pr_{i,j \in F} \left[ r_i(j) \ne Q(i,j) \mbox{ or } c_j(i) \ne Q(i,j)
\right] \leq 1/2 -\epsilon.$$
\end{lemma}
We first strengthen the conclusion obtained above slightly.
\begin{lemma}
\label{strengthen}
Let $\epsilon \geq {d/|F|}$ and
Let $r_i$ and $c_j$ be families of degree $d$ polynomials such that
$$\Pr_{i,j \in F} \left[ r_i(j) \ne c_j(i) \right] \leq 1/4 - \epsilon.$$
Then there exists a bivariate polynomial $Q$ of degree $d$ in each variable
such that
$$\Pr_{i \in F} \left[ r_i(\cdot) \ne Q(i,\cdot) \right] \leq 1/4$$
$$\mbox{ and } \Pr_{j \in F} \left[ c_j(\cdot) \ne  Q(\cdot,j)\right] \leq 1/4.$$
\end{lemma}
\begin{proof}
This lemma follows in a straightforward manner from Lemma~\ref{ASPS}. Let $Q$
be the bivariate polynomial guaranteed by Lemma~\ref{ASPS}. We define the
{\em bad rows} and {\em bad columns} as follows. Let
$$B_{\rm row} = \{i \in F | r_i(\cdot) \ne Q(i,\cdot)\} \mbox{ and let $x =
|B_{\rm row}|/|F|$.}$$
Similarly let
$$B_{\rm col} = \{j \in F | c_j(\cdot) \ne Q(\cdot,j)\} \mbox{ and let $y
= |B_{\rm col}|/|F|$.}$$
We count the number of points in $B_{\rm row} \times (F - B_{\rm col})$ which
satisfy $r_i(j) \ne c_j(i)$. For each bad row $i$ , there are at most $d$ points
for which $r_i(j) = Q(i,j)$. All the remaining points must lie on a bad
column or must satisfy $r_i(j) \ne c_j(i)$. Thus the fraction of violations in
any bad row (from the good columns) is at least $(1 - d/|F| - y)$. Similarly
we count the violations in bad columns and good rows and summing all theses
violations we get:
\begin{tabbing}
\=mm\=mm\=mm\=mm\=mm\=mm\=mm\=mm\=mm\kill
\> $ 1/4 - \epsilon $ \\
\> \> $\geq $ \> $ \Pr_{i,j} \left[ r_i(j) \ne c_j(i) \right] $ \\
\> \> $\geq $ \> $ ( \Pr_{i} \left[ i \in B_{\rm row} \right] $ \\
\> \> \>       $  * \Pr_{j} \left[ j \not\in B_{\rm col}
\mbox{ and } c_j(i) \ne r_i(j)  | i \in B_{\rm row} \right] )$ \\
\> \> \> $+ ( \Pr_{j} \left[ j \in B_{\rm col} \right] $\\
\> \> \> $ * \Pr_{i} \left[ i \not\in B_{\rm row}
\mbox{ and } c_j(i) \ne r_i(j)  | j \in B_{\rm col} \right] )$ \\
\> \> $\geq $ \> $x(1 - y - {d \over |F|}) + y(1 - x - {d \over |F|})$.
\end{tabbing}
We now use the fact that $x,y \leq 1/2$ and that $\epsilon \geq d/|F|$, to
reduce the above to
$x \leq 1/4$ and $y \leq 1/4$.
\end{proof}

We are now almost ready to prove Lemma~\ref{main lemma}. We first prove a
variant and then show how it implies the final result.
\begin{lemma}
\label{main lemma eq}
If $\delta_f \leq 1/8 - \epsilon$, then
for $x,h_1,h_2$ chosen uniformly at random from $F^m$,
$$\Pr_{x,h_1,h_2} \left [
\pfd_{x,h_1} (0) \ne \pfd_{x,h_2} (0)
\right ] \leq 4 \alpha \delta_f \mbox{ where $\alpha
= {4 \over \epsilon^2|F|}$}.$$
\end{lemma}
\remove{
{\em Note: To see that this suffices, observe that this implies that
$\Pr_{x,h_2} [ \corr_f(x) \ne \pfd_{x,h_2} (0) ] \leq \delta_f/2$.
In turn this implies that $\Pr_{x,h_2} [ \corr_f(x) \ne
P^{(\corr_f)}_{x,h_2} (0) ] \leq \delta_f/2$.
}
}
\begin{proof}
Pick $x,h_1,h_2,h_3$ at random from $F^m$ and consider
the set of points $\{x + ih_1 + jh_2 + ijh_3 | i,j \in F\}$.
We partition this set in two ways - by ``rows'' and by ``columns'' as
follows. For $i \in F$ let $\row_i = \{ x + ih_1 + jh_2 + ijh_3 | j
\in F\}$. Similarly for $j \in F$ let $\col_j = \{ x + ih_1 + jh_2 +
ijh_3 | i \in F\}$. Notice that each row and column is a line from the
space $F^m$. We first observe that these are actually random lines
(Here we call the distribution of lines picked by choosing a line
$l_{x,h}$ by picking $x,h \in F^m$ uniformly and randomly, to be the
{\em uniform} distribution over lines.)
\begin{claim}
For $i_1 \ne i_2 \in F$, the rows $\row_{i_1}$ and $\row_{i_2}$ are
independently and uniformly distributed over lines in $F^m$. (Similarly
for the columns.)
\end{claim}
Let $m(i,j) = f(x + ih_1 + jh_2 + ijh_3)$. Further let $r_i(\cdot) =
\pfd_{x + ih_1,h_2 + ih_3} (\cdot)$ and $c_j(\cdot) = \pfd_{x +
jh_2, h_1 + jh_3} (\cdot)$.
For a line $l_{x,h}$ from $F^m$, define $\delta(l_{x,h})$ to be
$\Pr_{t \in F} [ f(x + th) \ne \pfd_{x,h} (t)]$. Notice that $E_{x,h}
[\delta(l_{x,h})] = \delta_f$. The pairwise independence of the lines
implies that the collection of real numbers $\{\delta(\row_i)\}_{i \in
F}$ is a pairwise independent collection of variables taking values from
$[0,1]$ with expectation $\delta_f$. The second moment method thus allows
us to estimate the mean of this sample and shows that:
\begin{eqnarray}
\lefteqn{\Pr_{x,h_1,h_2,h_3} \left[
\sum_i \delta(\row_i) / |F| \geq 1/8 - \epsilon/2
\right] } \nonumber \\
& & \leq  \alpha \delta_f (1 - \delta_f) \mbox{ where }
\alpha = {4 \over \epsilon^2|F|}.
\label{first eqn}
\end{eqnarray}
A similar analysis applied to the columns yields:
\begin{eqnarray}
\lefteqn{\Pr_{x,h_1,h_2,h_3} \left[
\sum_j \delta(\col_j) / |F| \geq 1/8 - \epsilon/2
\right] }
\nonumber \\
& & \leq  \alpha \delta_f (1 - \delta_f) \mbox{ where }
\alpha = {4 \over \epsilon^2|F|}.
\label{second eqn}
\end{eqnarray}
By combining (\ref{first eqn}) and (\ref{second eqn}) yields
that with probability all but at most $2\alpha \delta_f$ over
four tuples $(x,h_1,h_2,h_3)$ we have,
$\Pr_{i,j \in F} [ r_i(j) \ne c_j(i) ] \leq 1/4 - \epsilon$. This allows
us to apply Lemma~\ref{strengthen} to claim that
for at least $3/4$ fraction of the $i$'s, $r_i(\cdot) \equiv
Q(i,\cdot)$ (and similarly for the columns).

Once again, we use pairwise independence  to show that
\begin{eqnarray}
\lefteqn{\Pr_{x,h_1,h_2,h_3} \left[
\{ i \in F | r_i(0) \ne m(i,0) \}\right.} \nonumber \\
& & \geq \left. (1/8 - \epsilon/2) |F|
\right] \nonumber \\
& \leq & \alpha \delta_f (1 - \delta_f) \mbox{ where }
\alpha = {4 \over \epsilon^2|F|}.
\label{third eqn}
\end{eqnarray}
\begin{eqnarray}
\lefteqn{\Pr_{x,h_1,h_2,h_3} \left[
\{ j \in F | c_j(0) \ne m(0,j) \} \right. } \nonumber \\
& & \geq \left. (1/8 - \epsilon/2) |F|
\right] \nonumber \\
& \leq  & \alpha \delta_f (1 - \delta_f) \mbox{ where }
\alpha = {4 \over \epsilon^2|F|}.
\label{fourth eqn}
\end{eqnarray}
Thus we now see that with probability at least $1 - 4 \alpha \delta_f$
all the events in (\ref{first eqn}), (\ref{second eqn}),
(\ref{third eqn}) and (\ref{fourth eqn}) hold.
In this case
$m(i,0) = Q(i,0)$ for at least $3/4 - 1/8 + \epsilon$
fraction of $i \in F$, which implies that $c_0(\cdot) = Q(\cdot,0)$.
Thus we have $P_{x,h_2}^{(f)} (\cdot) =
c_0 (\cdot) = Q(\cdot,0)$. Similarly $P_{x,h_1}^{(f)} (\cdot) = Q(0,\cdot)$.
Thus $P_{x,h_2}^{(f)} (0) = P_{x,h_1}^{(f)} (0) = Q(0,0)$.
\end{proof}
\begin{proof}[Lemma~\ref{main lemma}]
We start with the following observation:
\begin{eqnarray*}
\lefteqn{
\forall x, ~~~~ \Pr_{h_2} \left [
\plurality_{h_1} \{\pfd_{x,h_1}(0)\} \ne \pfd_{x,h_2}(0) \right] }
\\
& & \leq
\Pr_{h_1,h_2} \left [ \pfd_{x,h_1}(0) \ne \pfd_{x,h_2}(0) \right].
\end{eqnarray*}
We prove the above by running two different probabilistic experiments.
Say, a bag
has a number of colored balls, with the distribution of the number of balls
of each color being known. In the first game we nominate a color and
then pick a random ball and we lose if the color of the randomly chosen ball
is different from the nominated one. In the second game we pick two
balls (with replacement) at random from the bag and lose if the balls have
different colors. It is clear that in the first game the best choice is
to deterministically pick the most often occuring color in the bag, while the
second game corresponds to a mixed strategy for nominating the color in
the first game. Thus we are no more likely to lose in the first game than in the
second. The inequality above represents this analysis, with the
$h$'s corresponding to the balls and $\pfd_{x,h}(0)$'s
corresponding to their colors.

We now use the inequality above as follows:
\begin{eqnarray*}
\lefteqn{
E_{x} \left[ \Pr_{h_2} \left [
\begin{array}{r}
\corr_f(x)
= \plurality_{h_1} \{\pfd_{x,h_1}(0)\} \\
\ne \pfd_{x,h_2}(0)
\end{array}
\right] \right] } \\
& & \leq  E_{x} \left[ \Pr_{h_1,h_2} \left [
\pfd_{x,h_1}(0) \ne \pfd_{x,h_2}(0) \right]\right].
\end{eqnarray*}
In turn this implies
\begin{eqnarray*}
\delta_{\corr_f} & = &
\Pr_{x,h_2} \left[ \corr_f(x) \ne  \pfd_{x,h_2}(0) \right] \\
&\leq & \Pr_{x,h_1,h_2} \left[ \pfd_{x,h_1}(0) \ne \pfd_{x,h_2}(0)
\right].
\end{eqnarray*}
By Lemma~\ref{main lemma eq} the last quantity above is bounded by
$4\alpha\delta_f$. Thus if we choose $|F|$ to be sufficiently large
(strictly greater than $(16/\epsilon^2)$) then we get the conclusion
$\delta_{\corr_f} < \delta_f$.
\end{proof}

\section{Conclusions}
\label{next}

Here we list the two main consequences of Theorem~\ref{main theorem}.
The first is a straightforward corollary of the efficiency of the lines
test and talks about the local testability property of the
Polynomial-Line Codes (see Definition~\ref{defn:poly-line}).
\begin{theorem}
The Polynomial-Line Codes are $(2,1/8 -\epsilon)$ locally testable.
\end{theorem}
By applying Theorem~\ref{main theorem} to the task of constructing
efficient probabilistic verifiers, we get
small ``transparent'' proofs with low query complexity.
The transparent proofs so obtained are only slightly super-quadratic
($n^{2 + \epsilon}$-sized - where $n$ is the size of traditional proof)
in the length of the traditional proofs and the verifier probes them in
at most 165 bits and always accepts correct proofs, while rejecting
incorrect theorems with probability $1/2$.
To be able to lay out precise bounds
on the size of the proof, one needs to be careful about the model of
computing used to define the size of a proof. The model we use here is
the same as that used by \cite{PS}. In fact our verifier uses theirs
as a black box and then builds upon it.
In addition to the use of such size-efficient proof systems
our construction also use many ingredients from the query-efficient
proofs of \cite{BGLR}. The recursion mechanism of \cite{AS} plays a
central role in the combination of the various proof systems used here.
The final ingredient in the proof system is the randomness-efficient
parallelization protocol of \cite{ALMSS} (which is where the efficiency
of the tester of \cite{RS} plays a role).
Details of the construction
will be available in the full paper.

Last we would also like to mention two interesting questions
that may be  raised about locally checkable codes.
\begin{enumerate}
\item Does there exist a family of good $(2,1/2)$ locally-checkable
codes?
\item Does there exist such a family of codes with constant alphabet
size?
\end{enumerate}

\section*{Acknowledgments}

We would like to thank Laci Babai and Oded Goldreich for their valuable
comments. We also thank Steven Phillips for
providing us with a copy of the manuscript \cite{PhSa}.

\end{document}